\documentclass[prx,aps,twocolumn,superscriptaddress,groupedaddress,nofootinbib]{revtex4-1}
\usepackage{amsmath,amsfonts,amsthm,amssymb,bbm,graphicx,color,enumerate,hyperref,xcolor,tikz,pgfplots,ulem,subfigure}
\usepackage{algpseudocode}
\usepackage{algorithm}
\usepackage{MnSymbol}

\hypersetup{
	colorlinks=true,  
	linkcolor=blue,   
	citecolor=blue,   
	urlcolor=blue     
}

\def\C{ \mathcal{C} }

\def\E{ {\mathcal E} }

\def\M{ {\mathcal M} }

\def\U{ {\mathcal U} }
\def\R{ {\mathcal R} }

\def\N{ {\mathcal N} }

\def\O{ {\mathcal O} }

\def\D{ {\mathcal D} }

\newcommand{\ptm}[1]{\mathbf{#1}}
\def\Tr{{\rm{tr}}}
\def\tr{ \mbox{tr} }

\def\>{\rangle}
\def\<{\langle}

\def\hc{^{\dagger}}

\def\thv{\boldsymbol{\theta}}
\def\omv{\boldsymbol{\omega}}

\renewcommand{\emph}{\textit}

\newtheorem{theorem}{Theorem}

\newtheorem{lemma}[theorem]{Lemma}

\newtheorem{corollary}[theorem]{Corollary}

\begin{document}
\title{Classical simulations of noisy variational quantum circuits}
\author{Enrico Fontana$^{1,2,3}$}
\email{enrico.fontana@strath.ac.uk}
\author{Manuel S. Rudolph$^4$}
\author{Ross Duncan$^2$}
\author{Ivan Rungger$^3$}
\author{Cristina C\^{i}rstoiu$^{2}$}

\affiliation{$^1$Department of Computer and Information Sciences, University of Strathclyde, 26 Richmond Street, Glasgow G1 1XH, UK}
\affiliation{$^2$Quantinuum, Terrington House, 13-15 Hills Road, Cambridge CB2 1NL, UK}
\affiliation{$^3$National Physical Laboratory, Hampton Road, Teddington TW11 0LW, UK}
\affiliation{$^4$Institute of Physics, Ecole Polytechnique Fédérale de Lausanne (EPFL), CH-1015 Lausanne, Switzerland}

\begin{abstract}
Noise detrimentally affects quantum computations so that they not only become less accurate but also easier to simulate classically as systems scale up.  We construct a classical simulation algorithm, \textsc{lowesa} (low weight efficient simulation algorithm), for estimating expectation values of noisy parameterised quantum circuits. It combines previous results on spectral analysis of parameterised circuits with Pauli back-propagation and recent ideas for simulations of noisy random circuits. We show, under some conditions on the circuits and mild assumptions on the noise, that \textsc{lowesa} gives an efficient, polynomial algorithm in the number of qubits (and depth), with approximation error that vanishes exponentially in the physical error rate and a controllable cut-off parameter. We also discuss the practical limitations of the method for circuit classes with correlated parameters and its scaling with decreasing error rates.
\end{abstract}

\maketitle

\section{Introduction}

Quantum hardware has rapidly progressed to enable experiments that reach the barrier where computations become increasingly challenging to simulate with (high performance) classical computing systems~\cite{arute2019quantum, wu2021strong, zhong2020quantum, zhu2022quantum, moses2023race}.  Quantum advantage demonstrations recently stimulated substantial advances on classical algorithms for random circuit sampling~\cite{gao2018efficient, aharonov2022polynomial}, particularly on approximate tensor networks~\cite{xu2023herculean, huang2020classical-sup, villalonga2019flexible}.

In the current quest for applications with suitable implementations on noisy quantum hardware~\cite{preskill2018quantum}, much recent attention has been devoted to parameterised quantum circuits (PQCs). Variational quantum algorithms (VQAs) combine such a controllable quantum routine with classical optimisation to minimise a cost function that encodes the problem of interest. While they are often considered to have an intrinsic noise-resilience, recent studies have shown that accumulation of errors~\cite{gonzalez2022error} leads to phenomena like noise-induced barren plateaus~\cite{wang2021noise}, which hinder and potentially prohibit optimisation~\cite{anschuetz2022quantum}. Furthermore, frameworks~\cite{stilck2021limitations, de2023limitations} comparing classical algorithms with noisy VQAs by use of entropic quantities concluded that the circuit depth must have a bound that scales inversely with the physical gate error rate. Beyond this regime, classical methods certifiably outperform the noisy quantum computation~\cite{francca2022game}. Other classical simulations that target noisy VQAs include decision diagrams~\cite{huang2021logical} and tensor networks~\cite{ayral2023density, zhou2020limits}.  However, these approaches tend to be heuristic and do not necessarily provide rigorous trade-offs between complexity, approximation error and physical noise. 

Here we present \textsc{lowesa}, an efficient classical algorithm for simulating expectation values of parameterised quantum circuits affected by Pauli noise.
We combine ideas from simulating noisy random circuit sampling~\cite{aharonov2022polynomial} with spectral decompositions of parameterised noisy quantum circuits~\cite{fontana2022spectral}. Several applications, particularly in quantum machine learning~\cite{schuld2021effect, schreiber2022classical} have used the fact that cost functions for VQAs decompose into (finite) Fourier series in the variational parameters. Our previous work~\cite{fontana2022spectral, cirstoiu2017global} also shows that the effect of noise on the circuit produces Fourier coefficients that are contracted by a factor that decays exponentially with the Hamming weight of the frequency vector $\omv$. The algorithm we propose here produces an approximation to the noisy cost function that consists only of those Fourier modes with frequencies below a fixed cut-off value $|\omv|\leq \ell$. We show the time complexity of \textsc{lowesa} is $O(n^2m 2^\ell)$ for a \emph{specific class} of circuits on $n$ qubits consisting only of $m$ independently parameterised non-Clifford gates and any number of Clifford gates. The approximation error is shown to decay exponentially with $\ell$ and the physical gate error rate $p$, under mild assumptions on the noise. Equivalently, our algorithm takes $O(n^2 m (\frac{1}{\epsilon})^{1/p}))$ time to produce a function that approximates the noisy cost function within a fixed additive error $\epsilon$ (averaged over the entire parameter space). Notably, in the noise-less setting, simulating a circuit with $m$ non-Clifford gates, as considered here, require $O(\exp(m))$ \cite{bravyi2019simulation}. Improved sub-exponential algorithms for estimation of expectation values assume constant depth and planar architectures~\cite{bravyi2021classical}. By contrast, not only can we get a linear scaling in $m$ in the noisy setting but also, we recover an approximation of the entire cost function landscape rather than a single observable expectation value for a circuit with fixed parameters.    

It is important to emphasize that while \textsc{lowesa} is asympotically efficient in qubit number, the exponent scaling with $1/p$ can limit the practical runtime.  We leave for further research to investigate the extent to which the algorithm is computationally tractable for finite system sizes and low error rates of $p\approx 10^{-2} - 10^{-3}$ attained by current devices \cite{moses2023race}. On the other hand, approximate MPS-based tensor network simulation methods~\cite{ayral2023density} can also deal with (shallow) noisy circuits of large sizes up to hundreds of qubits, with approximation error that increases significantly with the gate fidelity. However, the complexity has exponential scaling with depth and more intricate circuit topologies beyond 1D. Furthermore, the relation between tensor truncation error and noise becomes difficult to quantify mathematically. It remains an interesting open question if our algorithm can be combined with such tensor network methods to improve the systems sizes accessible via (noisy) classical simulations.

The classical simulation approach presented here does not rely on the specific classical optimisation loop in VQAs, and therefore can be adapted to any algorithm that involves a class of noisy circuits with a fixed structure and independently parameterised gates. For example, certain noisy implementations of quantum signal processing~\cite{martyn2021grand} might fall under this.

Finally, our algorithm quantitatively reinforces the idea that gate fidelities of quantum devices need to decrease in order to access regimes beyond classical simulation methods~\cite{zhou2020limits}. Increasing the number of qubits for fixed error rates is unlikely to be sufficient as several noisy classical simulation algorithms exhibit polynomial scaling with qubit number, a recurrent feature that also appears in the case of noisy random circuit sampling~\cite{aharonov2022polynomial, gao2018efficient} and tensor network methods~\cite{ayral2023density}.  

\begin{figure*}
    \centering
    \includegraphics[width=0.9\textwidth]{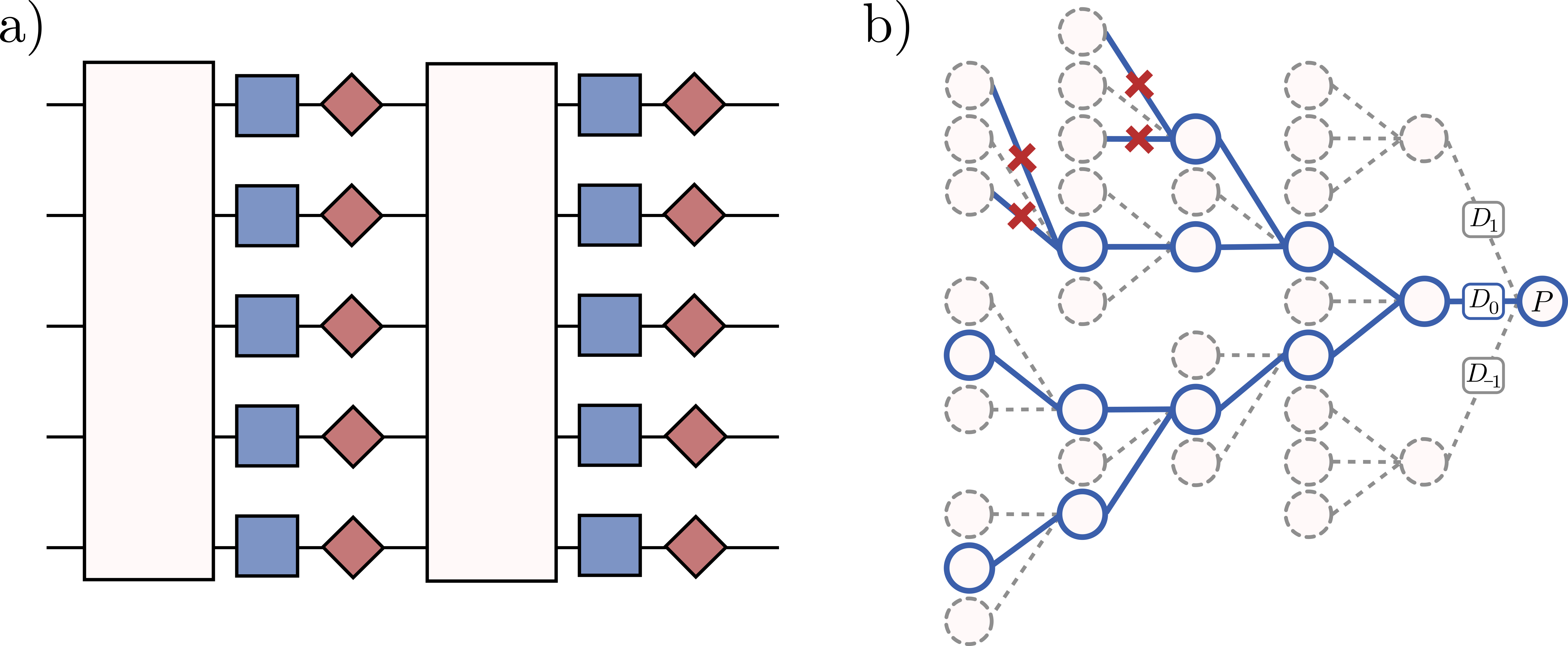}
    \caption{a) Schematic of the parameterised quantum circuits that can be simulated by \textsc{lowesa}. The light boxes are arbitrary (noisy) Clifford gates, the blue boxes are parameterised $z$-rotations and the red kites represent Pauli noise channels.\\
    b) Diagrammatic sketch of \textsc{lowesa} as described in Algorithm~\ref{algo:1} applied to circuits given by Equation~\eqref{eq:cliff_var_circ}. The Pauli operator $P$ is propagated backwards through the circuit where every Clifford gate transforms it into another Pauli, and the decomposition of the parameterised $Rz$ rotations into process modes $D_0, D_1, D_{-1}$ splits the propagation up into paths that may annihilate. A cut-off of $\ell=2$ is chosen which artificially annihilates paths that branch into $D_1,D_{-1}$ more than $2$ times.}
    \label{fig:lowesa-schematic}
\end{figure*}

\section{Background}

\subsection{Parameterised quantum circuits}
A PQC on $n$ qubits is defined as a sequence of $m$ unitary gates, each parameterised by a component of a parameter vector $\thv$.
Here, we consider the case where the gates are alternating layers of Clifford operations $C_i$ and single-qubit $z$-rotations $R^{(q_i)}_z(\theta_i)= e^{-i \theta_i/2 \, Z^{(q_i)}}$ such as
\begin{equation}
    \label{eq:cliff_var_circ}
	U(\thv) = \left( \prod_{i = 1}^{m} C_i R^{(q_i)}_z(\theta_i) \right) C_0 \,.
\end{equation}
The parameters $\thv \in [0, 2\pi]^m$ can therefore be equivalently described as rotation angles.
This specific form is operationally relevant, and since Clifford unitaries and single qubit rotations form a universal gate set, any PQC can be cast in this way (up to fixing a subset of the parameters).

Typical VQAs involve initialising the quantum computer in $|\boldsymbol 0\> = |0\>^{\otimes n}$, applying the PQC and measuring an observable to obtain a cost function. We denote the set of single-qubit Pauli operators by $\mathbb{P} =\{I, X, Y, Z\}$ and the expectation value for a specific $n$-qubit Pauli operator $P \in \mathbb{P}^{\otimes n}$ by 
\begin{equation}
	f(\thv) := \Tr(P \, \U_{\thv}[|\boldsymbol 0\>\<\boldsymbol 0|]),
    \label{eq:cost_fn}
\end{equation}
where the unitary channel is $\U_{\thv}[\cdot] := U(\thv)[ \cdot ] U^\dagger(\thv)$.

\subsection{Modelling noisy operations}
We are interested in the classical simulatability of VQAs affected by noise, and we model the noisy PQC using general Pauli channels, which are probabilistic mixtures of unitary $n$-qubit Pauli operator evolutions.  For a single qubit, a general Pauli channel is given by
\begin{align}
    \N_{Pauli}(p_X, p_Y, p_Z) [\rho] = &(1-p_X -p_Y-p_Z) \rho\\ \nonumber
        &+ p_X X\rho X + p_Y Y\rho Y + p_Z Z\rho Z.
\end{align}
These are often used to model local decoherent processes in quantum hardware. The dephasing channel $\N_{Pauli}(0, 0, p)$ is a particular example which models interactions between a qubit and the external environment.  The best-fit noise parameters $\{p_X, p_Y, p_Z\}$ for each qubit can be estimated experimentally via procedures like cycle benchmarking \cite{erhard2019characterizing}.

The general noise model we consider takes the form
\begin{equation}
    \tilde{\U}_{\thv} = \left( \bigcirc_{i}\; \tilde{\C}_i \circ \tilde{\mathcal{R}}^{(q_i)}_z(\theta_i) \right) \circ \tilde{\C_0} \label{eq:tildeU},
\end{equation}
where each single and two-qubit noisy gate is given by the target unitary followed by a Pauli channel acting on the same subset of qubits. Specifically, we have $\tilde{\R}_{z}^{(q_i)} (\theta_i) = \N_{Pauli}\circ \R_{z}^{(q_i)} (\theta_i) $ and $\tilde{\C}_i = \C_i\circ \M_i $, where $\M_i$ are multi-qubit Pauli channels. 

\subsection{Pauli transfer matrices and simulation algorithms}
\label{sec:simulationintro}
When studying generic quantum operations it can often be useful to use the \textit{Pauli transfer matrix} (PTM) formalism \cite{chow2012universal}. Let us briefly review it. 
In the PTM formalism, one takes the view of the normalised Pauli basis $\hat{\mathbb{P}} = \frac{1}{\sqrt{2}}\{I,X,Y,Z\}$, where a normalised Pauli operator $\hat{P}_i \in \hat{\mathbb{P}}^{\otimes n}$ is a basis vector $|P_i\rrangle$ in the space $\mathbb{R}^{4^n}$. The normalisation ensures that $\llangle P_i|P_j \rrangle = \tr(\hat P_i \hat P_j) = \delta_{ij}$. Quantum states can be written in this basis as $|\rho\rrangle$, 
\begin{equation}
    [|\rho\rrangle]_i = \tr(\rho \hat P_i),
\end{equation}
extending the identification of a one-qubit density matrix with its Bloch vector to higher dimensions. For instance, in this basis we represent the density matrix $|0\rangle\langle0|$ as $|0\rrangle = [1/\sqrt{2}, 0, 0, 1/\sqrt{2}]$.
Then, a quantum channel $\E$ is a matrix (the PTM) $\mathbf{E} \in \mathbb{R}^{4^n \times 4^n}$, 
\begin{equation}
    [\mathbf{E}]_{ij} = \llangle P_i | \mathbf{E} | P_j \rrangle = \tr(\hat P_i\E[\hat P_j]),
\end{equation}
and therefore expectation values of Pauli operators are written as $\llangle P_i | \mathbf{E} | \rho \rrangle = \tr(\hat P_i \E[\rho])$. Composition of quantum channels becomes matrix multiplication.

The PTM formalism can be used to calculate expectation values in the \textit{Heisenberg picture} via Pauli back-propagation, where the quantum channels are seen as acting on the measurement operator instead of the state~\cite{gottesman1998heisenberg}. In PTM form this adjoint operation corresponds to simply taking the transpose of the expectation value
\begin{equation}
    \llangle P | \mathbf{E} | \rho \rrangle = \llangle \rho | \mathbf{E}^\intercal | P \rrangle,
\end{equation}
which is possible for any $\E$. This perspective provides an efficient approach to classically computing expectation values. Take an $n$-qubit channel $\E$ and assume it can be decomposed as a sum of $N$ Clifford unitary channels $\E_i$ via $\E = \sum^N_{i=1} c_i \E_i$, $\sum_{i} c_i = 1$. Also consider a stabiliser state~\cite{gottesman1998heisenberg} $\rho$ such that the expectation value with any Pauli operator can be evaluated efficiently. Then, given a Pauli $P$, the expectation value $\llangle P | \mathbf{E} | \rho \rrangle$ can be expanded as a sum of $N$ terms $\llangle P | \mathbf{E}_i | \rho \rrangle$. As Clifford unitaries are generalised permutation matrices in the PTM representation we get $\llangle \rho | \mathbf{E}^\intercal_i | P \rrangle = \llangle \rho | P'_i \rrangle$ (up to a phase), for some other Pauli operator $P_i'$. When $\E_i$ is an $n$-qubit Clifford unitary then it can be synthesised into at most $O(n^2/log(n))$ gates \cite{aaronson2004improved} and the change of Pauli frame from $P$ to $P_i'$ can be efficiently computed in $O(n^2)$ \cite{chamberland2018fault, gottesman1998heisenberg}. Finally, since $\rho$ is assumed a stabiliser state, the expectation value $\llangle \rho | \mathbf{E}^\intercal_i | P \rrangle$ can be efficiently computed in $O(n^2)$. This gives an efficient classical algorithm to compute expectation values when $N \sim \text{poly}(n)$. 

\subsection{Prior art}
This approach is not new. The decomposition of general channels into sums of stabiliser channels (Cliffords and Pauli measurements) for the purpose of quantum circuit simulation was introduced in Ref.~\cite{bennink2017unbiased}. A similar sum-over-Clifford algorithm for unitary circuits was explored in Ref.~\cite{bravyi2019simulation}. A PTM-based algorithm for both exact and noisy circuit simulation has been proposed in Ref.~\cite{huang2022classical} from a Schr\"odinger perspective (state propagation).
The work in Ref.~\cite{rall2019simulation} is the closest to the method used here, as it covers the PTM representation in conjunction with a Heisenberg picture simulation method. In addition, it discusses the effect on simulatability of adding symmetric depolarising noise on $z$-rotation gates.

However, something that to our knowledge has not been made explicit before is that the method can be generalised beyond decompositions into Clifford unitaries (or near-Clifford unitaries \cite{bravyi2019simulation}) and Pauli measurement channels.  Indeed, here we will consider general processes $\E_i$ for which the expectation value $\llangle P | \mathbf{E}_i | \rho \rrangle$ can be evaluated efficiently. Notably, the processes $\E_i$ need not even be valid quantum channels (or completely positive trace preserving maps), we only require that its PTM representation is sufficiently sparse. This occurs when the adjoint channel $\E_i^{\dagger}$ maps every Pauli operator into a combination of small, $O(\text{poly}(n))$, number of Pauli operators. This echoes remarks in Ref.~\cite{nest2009simulating}, although that work is in the Schr\"odinger picture. In our case, the $\E_i$ will correspond to compositions of Clifford unitaries and processes that map every Pauli operator to a single (possibly distinct) Pauli operator or to zero.

\section{Classical simulation of uncorrelated parameter VQAs}
\label{sec:Simulation}

\subsection{Strategy}
We first show how the noisy variational circuits considered in Equation~\eqref{eq:cliff_var_circ} admit a linear decomposition into processes that are amenable to the classical simulation outlined in Sec.~\ref{sec:simulationintro}.  To that aim, it turns out that a decomposition into Fourier series of the noisy channel $\tilde{\U}_{\theta}$, and therefore noisy cost function, results in processes that map a Pauli operator into multiple Pauli operators, and thus their composition may lead to an exponential accumulation of terms (see Sec.~\ref{sec:example}).  However, a different choice of basis involving trigonometric polynomials remedies this to produce a decomposition for which the dominant coefficients in the expansion can be efficiently computed.

Let $\mathbf{R}_z(\theta)$ be the PTM of $\mathcal{R}_{z}(\theta)$ and let $\mathbf{N}$ be the PTM of the Pauli noise channel $\N_{Pauli}$, $\mathbf{N} = \text{diag}(1, q_X,q_Y, q_Z)$. The eigenvalues of the Pauli channel are related to the error probabilities as $q_X = 1- 2(p_Z+p_Y)$, $q_Y = 1- 2(p_Z+p_X)$, $q_Z = 1- 2(p_X+p_Y)$. 
Then, the noisy channel $\tilde{\R}_z(\theta) = \N_{Pauli} \circ \R_z(\theta)$ has, with respect to the orthonormal basis $\{|I\rrangle, |X\rrangle, |Y\rrangle, |Z\rrangle \}$, the PTM
\begin{equation}
    \mathbf{N\cdot R} = \begin{pmatrix}
	1 & 0 & 0 & 0\\
	0 & q_X \cos{\theta}& - q_X\sin{\theta} & 0\\
	0 & q_Y\sin{\theta} & q_Y\cos{\theta} & 
        0\\
	0 & 0 & 0 & q_Z
 \end{pmatrix}\,.
\end{equation}
Denote the projectors by $\Pi_0 =|I\rrangle\llangle I| + |Z\rrangle\llangle Z|$, $\Pi_X = |X\rrangle\llangle X|$ and $\Pi_Y = |Y\rrangle \llangle Y|$.
Then we can define new quantum processes $\{\mathcal{D}_0, \mathcal{D}_{1}, \mathcal{D}_{-1}\}$ to be used in the simulation algorithm via their PTM representation ${\mathbf{D_0}} = \Pi_0 {\mathbf{NR} }\Pi_0$, $\mathbf{D_1} = \Pi_X {\mathbf{NR}} \Pi_X + \Pi_Y {\mathbf{NR}} \Pi_Y $ and $\mathbf{D_{-1}} = \Pi_X {\mathbf{NR}} \Pi_{Y} + \Pi_Y {\mathbf{NR}} \Pi_{X} $ such that
\begin{equation}
    \ptm{N\cdot R} = \ptm{D_0} + \cos{\theta}\, \ptm{D_1} + \sin{\theta}\, \ptm{D_{-1}}.
\end{equation}
Expanding out these processes, we see that each of them maps any single Pauli operator into at most another single Pauli operator (up to a scaling),
\begin{align}
\label{eq:proc_modes}
\mathbf{D}_0 = \begin{pmatrix}
	1 & 0 & 0 & 0\\
	0 & 0 & 0 & 0\\
	0 & 0 & 0 & 0\\
	0 & 0 & 0 & q_Z
\end{pmatrix},& \;\;
\mathbf{D}_1 = \begin{pmatrix}
	0 & 0 & 0 & 0\\
	0 & q_X & 0 & 0\\
	0 & 0 & q_Y & 0\\
	0 & 0 & 0 & 0
\end{pmatrix},\\
\mathbf{D}_{-1} = &\begin{pmatrix}
	0 & 0 & 0 & 0\\
	0 & 0 & -q_X & 0\\
	0 & q_Y & 0 & 0\\
	0 & 0 & 0 & 0
\end{pmatrix}.
\end{align} 

This decomposition allows us to expand the noisy circuits in terms of a multivariate trigonometric basis, which is a more convenient choice for the classical simulation. Consider $\Phi_{\omv}(\thv) := \prod_{i=1}^m \phi_{\omega_{i}} (\theta_{i})$ where $\phi_0(\theta) = 1,\; \phi_1(\theta) = \cos(\theta), \; \phi_{-1}(\theta) = \sin(\theta)$ are \textit{trigonometric monomials} that encode the $\thv$ dependence. Then, the \emph{noisy} variational circuits admit the decomposition 
\begin{equation}
    \tilde{\U}_{\theta} = \sum_{\omv \in \{0,\pm 1\}^{m}} \Phi_{\omv}(\thv) \mathcal{D}_{\omv},
\end{equation}
where each process $\D_{\omv}$ is labelled by a frequency vector $\omv \in  [0, \pm 1]^m$ and given by
$\mathcal{D}_{\omv} := \left(\bigcirc_{i}\; \mathcal{C}_i \circ \mathcal{D}_{\omega_i} \right)  \circ \mathcal{C}_0 $. In keeping with previous work~\cite{cirstoiu2017global} we call these channels \emph{process modes}.

Overall, this new decomposition yields the following Fourier series representation for the cost function in Equation~\eqref{eq:cost_fn}
\begin{equation}\label{eq:cost}
	\tilde{f}(\thv) = \sum_{\omv} d_{\omv} \Phi_{\omv}(\thv).
\end{equation}
The \emph{Fourier coefficients} are given by
\begin{equation}
    d_{\omv} := \tr(P \D_{\omv}[|0\rangle\langle0|]) = \sqrt{2^n} \llangle P | \mathbf{D}_{\omv} | 0 \rrangle,
\end{equation}
where the factor $\sqrt{2^n}$ is necessary since we have defined $f(\thv)$ as the expectation value of an unnormalised Pauli operator.

Note that in the above, the Clifford unitaries $\C_i$ were noise-free and the parameterised rotation gates carried a time-independent Pauli noise. A similar decomposition arises when we consider the general Pauli noise model for $\tilde{\C_i} = \C_i \circ \M_i$. In this case, we denote the resulting process modes by $ \mathcal{D}'_{\omv} := \left(\bigcirc_{i}\; \mathcal{C}_i \circ \M_i \circ \mathcal{D}_{\omega_i} \right)  \circ \mathcal{C }_0 \circ \M_0$ and the corresponding coefficients by $d'_{\omv} = \sqrt{2^n} \llangle P | \mathbf{D'}_{\omv} | 0 \rrangle $. We first describe and analyse the proposed classical algorithm for the simpler noise model that only affects the parameterised gates. This is purely to make the exposition easier to follow. The same principle works in the general case (see Sec.~\ref{sec:generalnoise}). Furthermore, the analysis extends to time-dependent Pauli errors. The noise models considered here also include the local depolarising channels that have been previously used in classical algorithms for noisy random circuit sampling~\cite{aharonov2022polynomial}. Both in our case and in previous work there is an implicit assumption that the Pauli error probabilities for each gate are known a-priori.

\subsection{The LOWESA simulation algorithm}

We are now in a position to state the simulation algorithm, which shares similar features to the algorithm in Ref.~\cite{gao2018efficient}, but applied to the task of estimating expectation values and to a different family of circuits.
We name it \textsc{lowesa} for \textsc{lo}w \textsc{w}eight \textsc{e}fficient \textsc{s}imulation \textsc{a}lgorithm (pronounced ``low-EE-sa''). 

Given a cut-off parameter $\ell$, \textsc{lowesa} returns a \emph{function} $\tilde{g}$ approximating the noisy cost function $\tilde{f}$ constructed from all the low-weight $|\omv|\leq \ell$ terms. 
This function is expressed as a trigonometric series and can therefore be used to evaluate the cost estimate for any parameter vector $\thv$ using
\begin{equation}
    \tilde{g}(\thv) = \sum_{|\omv|\leq \ell} d_{\omv} \Phi_{\omv}(\thv)
\end{equation} 
with low computational effort.
As the algorithm produces all $\{d_{\omv}
\}_{|\omv|\leq l}$, the function evaluation is independent of qubit number and depth.

\clearpage
\textsc{lowesa} involves the following steps:
\begin{algorithm}[H]
\caption{ [{\textsc{\bf LOWESA}}] Simulating cost functions of noisy VQAs with uncorrelated angles }
\label{algo:1}
\raggedright\textbf{Input:} Quantum circuit given by Equation~\eqref{eq:cliff_var_circ} affected by Pauli noise with probabilities $p_Z > 0$ or $p := \min\{p_{X/Y}\} > 0$, and defined by process modes $\{\D_{\omv}\}$; measurement Pauli operator $P$; cut-off parameter $\ell$.\\
\textbf{Output:} $\tilde{g}(\thv)$, an approximation of $\tilde{f}(\thv)$.
\begin{algorithmic}[1]
\Procedure{lowesa}{}
\State $\tilde{g}(\thv) \leftarrow 0$
\State {\bf{Run}} \textsc{Low-weight coefficients}  to calculate $d_{\omv} = \sqrt{2^n} \llangle 0| \mathbf{D}^\intercal_{\omv} |P\rrangle$  $\forall \ |\omv|\le\ell$.
\ForAll{$|\omv|\le \ell$ with non-zero $d_{\omv}$}
\State $\tilde{g}(\thv) \leftarrow \tilde{g}(\thv) + d_{\omv} \Phi_{\omv}(\thv)$
\EndFor
\State\textbf{Return} $\tilde{g}(\thv)$
\EndProcedure
\end{algorithmic}
\textbf{Subroutine:} Calculate non-zero $d_{\omv} = \sqrt{2^n} \llangle 0| \mathbf{D}^\intercal_{\omv} |P\rrangle$  $\forall \ |\omv|\le\ell$.
\begin{algorithmic}[1]
	\Procedure{Low-weight coefficients}{}
	\For{$i=m$ to $1$}{}
	\State $P\leftarrow  C_i\hc P C_i$, up to a phase $\phi_i \in \{\pm 1\}$
	\If{$P_{q_i} \in \{I, Z\}$  (i.e $P$ on qubit $q_i$)}{}
	\State $\omega_i = 0$
	\ElsIf{$P_{q_i} \in \{X, Y\}$ {\bf{and}}  $|\omv|  < l$ }{}
	\State Split into branches $\omega_i =1$ or $\omega_i =-1$
	\State $|\omv| \leftarrow + 1$ 
	\EndIf
	\State $d_{\omega_i}P \leftarrow \D_{\omega_i}\hc(P)$
	\EndFor
	\State $P\leftarrow C_0\hc P C_0$ and phase $\phi_0 \in \{\pm 1\}$
	\State\textbf{Return} $d_{\omv} = \sqrt{2^n} (\prod_{i=1}^{m}d_{\omega_i}\phi_i )\phi_0 \llangle 0|P\rrangle$
	\EndProcedure
\end{algorithmic}
\end{algorithm}
Note that, while our exposition deals with expectation values of a single Pauli operator, the results extend immediately to general observables that can be decomposed into ${\rm{poly}}(n)$ Pauli operators.

To analyse the asymptotic complexity, we need to (1) show that each term $d_{\omv}$ can be efficiently estimated, (2) count the number of non-zero process modes $\D_{\omv}$ with $|\omv|\leq \ell$, and (3) evaluate the accuracy in the approximation $\tilde{g}\approx \tilde{f}$. Condition (1) is satisfied by construction - the choice of trigonometric basis ensures that (the adjoint of) $\D_{\omv}$ maps a Pauli operator to either zero or a (different, scaled) Pauli operator. Each $d_{\omv}$ can be individually estimated in at most $O(n^2 m)$ steps using the Pauli back-propagation method outlined in Sec.~\ref{sec:simulationintro}.

For (2), note that while there are $\binom{m}{|\omega|} \, 2^{|\omega|}$ paths with a fixed weight $|\omega|$ for a total of at most $m^{\O(\ell)}$ within the cut-off, many of these will be zero when acted upon the input $|P\rrangle$. This is due to the fact that process modes in Equation~\eqref{eq:proc_modes} each annihilate half of the Paulis. A more efficient algorithm than calculating all possible process modes thus works as follows.

\begin{figure*}
    \centering
    \includegraphics[width=0.95\textwidth]{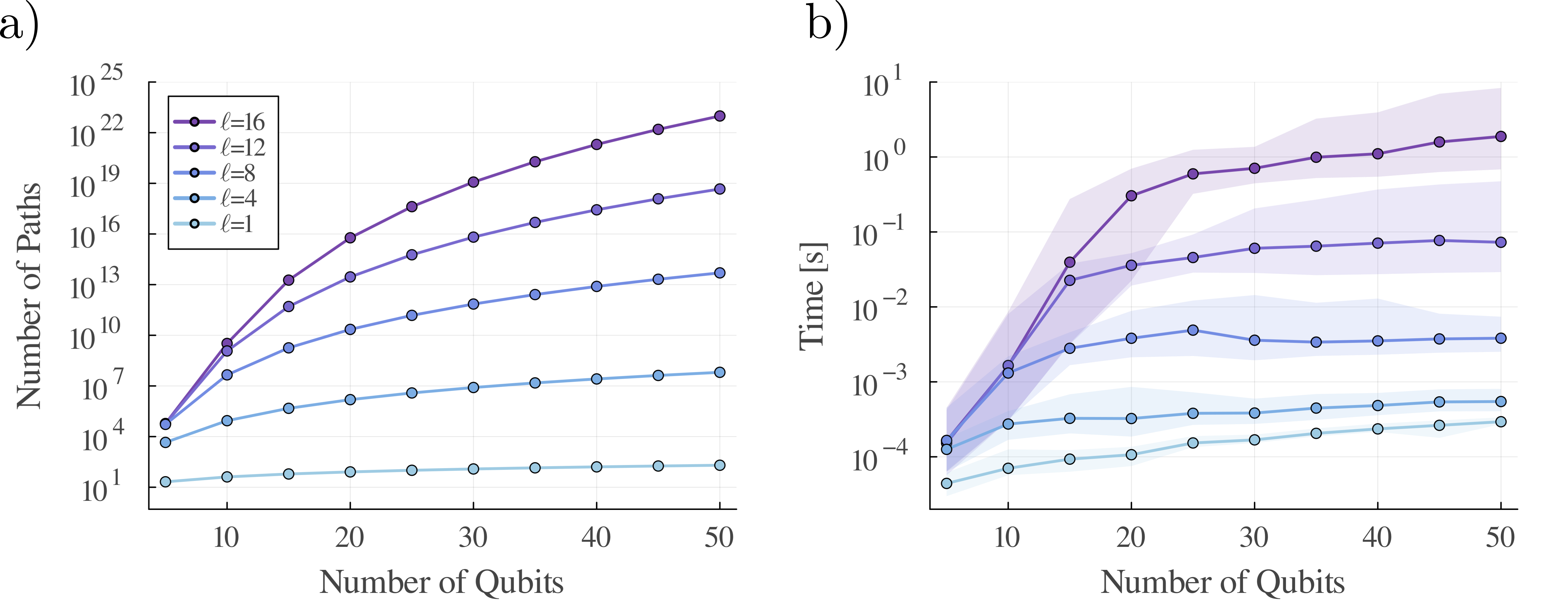}
    \caption{Scaling of \textsc{lowesa} with the number of qubits $n$ and cut-off parameter $\ell$. The circuit structure consists of two parameterised layers of $H-Rz(\theta_i)-X-H$ on each qubit, where the Hadamard and X gates are chosen with 0.5 probability, followed by CNOTs placed on a 2D topology. a) Total number of paths for a given $\ell$, which equals $\sum_i^{\ell} {m \choose i} 2^i$. Note that the number of paths that \textsc{lowesa} needs to explore is dramatically lower. 
    b) Wall time to run \textsc{lowesa} with truncation parameter $\ell$ on an average laptop without parallelisation.
    Each data point represents an average over 500 different randomized circuits with Pauli Z measurement operators that act on a random subset of qubits. 
    The shading shows the 90\% confidence interval. The simulation of the Clifford gates used a look-up table, meaning that the scaling in $n$ is entirely due the scaling of $m$ with $n$.
    }
    \label{fig:runtimes}
\end{figure*}

Start with the target Pauli measurement operator $P$ and propagate in the Heisenberg picture through the circuit. For each Clifford unitary $C_i$, updating the Pauli operator (by conjugation) takes at most $O(n^2)$. Each process $\D_{\omega_i}$ within a path $\D_{\omv}$ acts non-trivially on a single qubit, $q_i$. If the propagated Pauli operator on that qubit is either $I$ or $Z$  then only $\D_0$ leads to a non-zero path, otherwise if it is $X$ or $Y$ either $\D_{1}$ or $\D_{-1}$ are valid choices. As only $\D_{\pm 1}$ contribute to the total weight and $|\boldsymbol{w}|\leq \ell$, it suggests a binary tree-like data structure with $\ell$ layers to keep track of the change of Pauli frame and the different branching possibilities. A branch may terminate sooner than if it propagated the Pauli through the entire circuit. The number of branches and therefore valid paths $\D_{\omega}$ will be at most $2^{\ell}$. Putting everything together, this reduces the total complexity of evaluating all non-zero $d_{\omega}$ with $|\omega|\leq \ell$ to $O(n^2 m 2^\ell)$ in the worst case. We note that the quadratic scaling in $n$ is for general $n$-qubit Clifford unitaries, and can be improved for $k$-local (or sparse) unitaries. In particular, if one fixes the set of Clifford unitaries that are executed within the circuits (for example the set $\{X, H, CNOT\}$), one can employ time-memory trade-off tools like look-up tables for each $k$-body Clifford operation and how they act on every $k$-body Pauli operator.
In Figure~\ref{fig:runtimes} we illustrate the runtime of \textsc{lowesa} using this technique on a circuit structure that is typically challenging for classical simulators.

Finally, condition (3) remains to be verified so that \textsc{lowesa} yields an accurate simulation of the noisy cost function. Given a cost function $\tilde{f}$ and its approximation $\tilde{g}$, we define the average $L^2$-norm error over the space of parameters $\Theta = [0, 2\pi]^m$
\begin{equation}
    \label{eq:err}
	\Delta(\tilde{f}, \tilde{g}) := \left(\frac{1}{|\Theta|} \int_\Theta |\tilde{f}(\thv) - \tilde{g}(\thv)|^2 d\thv\right)^{1/2},
\end{equation}
where the integration measure is $d\boldsymbol{\theta} = d \, \theta_1 d\, \theta_2 ...d\, \theta_m$ and $|\Theta| = (2\pi)^{m}$ is a normalisation factor so that $\frac{1}{|\Theta|} \int d\, \boldsymbol{\theta} =1$.
In Appendix~\ref{ap:proof-main} we prove the following result
\begin{theorem}
    \label{thm:uncorr_cliff}
    Consider a $n$-qubit VQA with a PQC as in Equation~\eqref{eq:cliff_var_circ} having $m$ independently parameterised $z$-rotations affected 
    by a single-qubit Pauli noise channel $\mathcal{N}_{Pauli}(p_X, p_Y, p_Z)$ as in Equation~\eqref{eq:tildeU}.
    Define $p = \min\{p_X, p_Y\}$ and require at least one of $p$, $p_Z$ be greater than 0.
    
    Then, for any weight cut-off $\ell \in \mathbb{N}$, \textsc{lowesa} (Algorithm~\ref{algo:1}) returns an approximation $\tilde{g}$ for the noisy cost function $\tilde{f}$ with error
    \begin{equation}
        \Delta(\tilde{f}, \tilde{g}) \le (1 - 2p - 2p_Z)^{\ell + 1} \le e^{-2(p+p_Z)\ell}
    \end{equation}
    and runs in time at most $O(n^2 m \, 2^\ell)$.
\end{theorem}
 It follows from Theorem~\ref{thm:uncorr_cliff} that \textsc{lowesa} is efficient, as it scales polynomially with $m$ and $n$; however, the scaling with noise probability is considerably worse. For example, suppose we wish to have an error bounded by $\epsilon$. Then one would choose $\ell \approx \frac{1}{2p + 2p_Z} \log\epsilon^{-1}$, giving a runtime $O(2^{\frac{\log \epsilon^{-1}}{2p + 2p_Z}}\, n^2 m) $. While this is asymptotically efficient in the width and depth of the circuit, the dependency on the error rate limits its practicality. Notably the exponent may still be considerably large if the noise is small. When the goal is to simulate the expected outcome of a hardware implementation with a finite number of measurements $N_s$, the error can be chosen like $\epsilon\in\mathcal{O}(\frac{1}{\sqrt{N_s}})$, thus relaxing the precision requirements.

In Figure~\ref{fig:accuracy} we illustrate the mean accuracy of the algorithm for an example circuit of the hardware-efficient family. We observe that the error is typically up to two orders of magnitude lower than the bounds, suggesting these are loose and may be improved for the typical case.

\subsection{General Pauli noise models}
\label{sec:generalnoise}

The result can be extended to cover multi-qubit Pauli noise affecting all gates, not just the parameterised ones. In Appendix~\ref{ap:proof-extended} we prove the more general result:

\begin{theorem}\label{thm:general_noise}
Consider an $n$-qubit VQA under the noise model
\begin{equation}\label{eq:tildeU_general}
    \tilde{\U}_{\thv} = \left( \bigcirc_{i=1}^m\; \M_i \circ \C_i \circ \N_i \circ \mathcal{R}^{(q_i)}_z(\theta_i) \right) \circ \M_0 \circ \C_0
\end{equation}
where $\{\M_i\}$ are $n$-qubit Pauli channels with layer-dependent noise parameters and every z-rotation is followed by a local single-qubit Pauli noise $\N_i = \N_{Pauli}(p^i_X,p^i_Y,p^i_Z)$ with $p' = \min_i\{p^i_X,p^i_Y\}$, $p'_Z = \min_i\{p^i_Z\}$. Assume at least one of $p'$, $p'_Z$ is greater than 0.

Then, for any weight cut-off $\ell \in \mathbb{N}$, \textsc{lowesa} (Algorithm~\ref{algo:1}) with modified process modes $\mathcal{D}'_{\omv} = (\bigcirc_{i=1}^{m} \M_i \circ \C_i \circ \D_{\omega_i})\circ \C_0\circ \M_0$ and coefficients $d'_{\omv} = \sqrt{2^n} \llangle P| {\bf{D}'_{\omv}}|0\rrangle$ returns an approximation $\tilde{g}$ for the cost function $\tilde{f}$ with error
\begin{equation}
	\Delta(\tilde{f}, \tilde{g}) \le (1-2p'-2p'_Z)^{\ell+1}
    \le e^{-2(p'+p'_Z)\ell}
\end{equation}
and runs in time  at most $O(n^2 m 2^\ell)$.
\end{theorem}
The result relies on the fact that any Pauli channel will map a propagated Pauli operator to itself, up to a proportionality factor that can be at most 1. In other words, this means that each of the modified process modes $\D_{\omv}'$ will act similarly to the previously considered modes $\D_{\omv}$ arising from the simplified error model, so that ${\bf{D}'}_{\omv}|P\rrangle \propto {\bf{D}}_{\omv}|P\rrangle$. Therefore the proof and the bounds follow in the same way as for Theorem~\ref{thm:uncorr_cliff}. The only modification to the algorithm is that to compute $d_{\omv}'$ one must also keep track of these proportionality factors along with the propagated Pauli. 

For general noise, it is difficult to improve upon the upper bound on the approximation error $\Delta$ since one can be in a situation where along the paths of weight $|\omv|=\ell+1$ the proportionality factors might all be 1 when propagating the Pauli operator through each Pauli channel $\M_i$. In practical situations the coefficients $d'_{\omv}$ will contract. For instance, let's assume that in the decomposition of the $n$ qubit Clifford operator $\C_i$ into primitive (single and two-qubit) gates each incurs a local single-qubit depolarising channel $\N_{dep}$ with error probability $\eta$. Then it follows we can find a tighter bound
\begin{equation}
    \Delta(\tilde{f}, \tilde{g}) \leq (1-2p-2p_Z)^{\ell+1} (1-\eta)^{\ell+1}.
\end{equation}
This comes from the fact that $\N_{dep}^{\dagger}(P)= (1-\eta)P$ if $P\in \{X,Y,Z\}$ and $\N_{dep}^{\dagger}(I)= I$ along with the previous observation that for valid paths leading to non-zero coefficients, $\D_{\pm 1}$ are applied to qubit $q_i$ whenever the propagated Pauli on qubit $q_i$ is not $I$ or $Z$. Therefore the noise from the Clifford part will contribute and at the very least contract by a factor of $(1-\eta)$ whenever we have a branching possibility to apply either $\D_{+1}$ or $\D_{-1}$, which are the only contributors to the total weight $|\omv|$. Note that this type of noise model has previously been considered in the context of noisy random circuit sampling~\cite{aharonov2022polynomial}. 

\begin{figure}[t]
    \centering
    \includegraphics[width=0.48\textwidth]{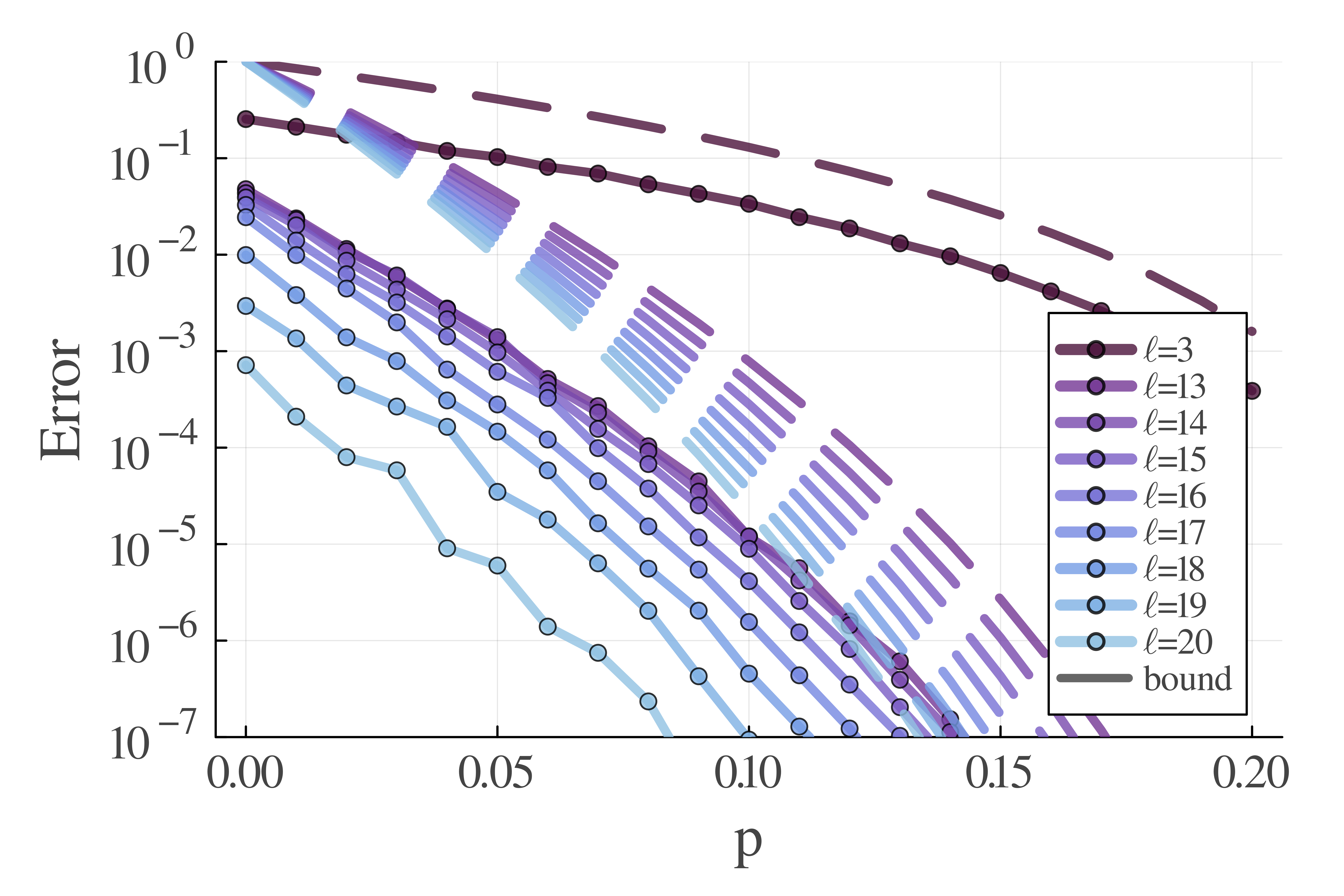}
    \caption{Accuracy benchmark of \textsc{lowesa} compared to the error bounds as predicted in Theorem~\ref{thm:uncorr_cliff}. We show the $L^2$ error of a single-qubit Pauli Y operator expectation with $\ell<m=60$ for two layers of a $n=10$ qubit circuit. The circuit consists of parametrized single-qubit gates $R_z(\thv_i)\,R_x(\thv_{i+1})\,R_z(\thv_{i+2})$ on each qubit followed by CNOT gates in a 2D topology. For this particular circuit, each entangler in the 2D topology was placed with a 0.5 probability. The noise model is symmetric depolarising noise, where the parameters are set $p_X = p_Y = p_Z = p$. Each point is averaged over 1000 random parameterisations of the same circuit to compare to the integral definition of our error bounds. All paths below $\ell=3$ and above $\ell=21$ annihilate. Consequently, the simulation with $\ell=21$ is exact.}
    \label{fig:accuracy}
\end{figure}

\subsection{Fixed (unparameterised) non-Clifford gates}
The extension of \textsc{lowesa} to the case where non-Clifford unparameterised gates are present is straightforward. As was done in Ref.~\cite{fontana2022efficient}, one may treat non-Clifford rotation gates as parameterised rotation gates that have their parameters fixed on at a later stage. A circuit with $t$ fixed $z$-rotation gates and $m$ parameterised $z$-rotation gates may be transformed into a circuit with $m+t$ $z$-rotations for simulation purposes, obtaining a cost function $F(\thv, \boldsymbol{\phi})$. Then the intended cost function is obtained by fixing $\boldsymbol{\phi}$. It follows that any statement on the simulation runtime still applies with the substitution $m \rightarrow m + t$.
However, getting an error bound with non-Clifford gates is more complicated, since we can no longer average over the expanded parameter space owing to the fixed gates. We can still make a weaker probabilistic statement, proven in Appendix~\ref{ap:proof-fixed}.
\begin{theorem}
    \label{thm:uncorr_random}
    Consider a variational circuit consisting of $m$ uncorrelated noisy parameterised rotation gates, and $t$ noisy rotation gates with fixed random angles independently and uniformly distributed. The noise model is that of Theorem~\ref{thm:general_noise}. Then for any $k \ge 1$ and weight cut-off $\ell \in \mathbb{N}$, the simulation error of \textsc{lowesa} (Algorithm~\ref{algo:1}) with modified process modes obeys
    \begin{equation}
        P\left(\Delta(\tilde{f}, \tilde{g}) \ge (1 + k)e^{-2(p'+p'_Z)\ell} \right) \le \frac{1}{k^2}
    \end{equation}
    and the Algorithm runs in time $O(n^2 (m+t) 2^\ell)$.
\end{theorem}
Theorem~\ref{thm:uncorr_random} implies that for a typical choice of the $\boldsymbol\phi$ angles the error is still exponentially suppressed in $\ell$. Alternatively, suppose one would like the error to be bounded by $\epsilon$ with probability $ \ge 1 - \delta$. Then one would choose $\ell \approx \frac{1}{2p' + 2p'_Z} (\log \epsilon^{-1} + \log(1+\frac{1}{\sqrt{\delta}}))$, giving a runtime which is only slightly worse than the one from the previous Theorems, for reasonable choices of $\delta$.

\section{The case of correlated parameters}

The main result has been derived assuming that the parameters controlling the rotation gates in the circuits are uncorrelated. One may therefore wonder whether it extends to correlated parameter circuits, which are ubiquitous in quantum machine learning~\cite{biamonte2017quantum} as well as forming the basis of algorithms like the Quantum Approximate Optimisation Algorithm (QAOA)~\cite{farhi2014quantum, hadfield2019quantum} or the Hamiltonian Variational Ansatz (HVA) for chemistry problems~\cite{wecker2015progress, kandala2017hardware}.

However, the argument used in the proof of Theorem~\ref{thm:uncorr_cliff} does not hold since with correlated angles the basis functions are no longer orthogonal over the correlated parameter space. For example, consider the following case where
\begin{align}
    &\Phi_2(\theta) = \cos^2(\theta), \;\;\; \Phi_{-2}(\theta) = \sin^2(\theta) \\
    &\Rightarrow\frac{1}{2\pi} \int \Phi_2(\theta) \Phi_{-2}(\theta) d\theta = \frac{1}{8} \neq 0
\end{align}
Indeed one can devise simple examples of correlated angles system where the bound is seen to fail. 
Consider the following 1-qubit correlated parameter circuit
\begin{equation}
    U_d(\theta) = H (R_z(\theta))^d H
\end{equation}
It is simple to show that when $U_d(\theta)$ is applied to the initial state $|0\>$, then measuring the Pauli Z produces the cost function
\begin{equation}
    f_d(\theta) = \cos(d\theta)
    = \sum_{i=0}^{\lfloor d/2 \rfloor} (-1)^{i} {d \choose 2i}\; \sin^{2i}(\theta) \cos^{d-2i}(\theta)
\end{equation}
whose terms are all of weight $d$. Therefore, any reconstruction with weight $\ell < d$ would trivially return $\tilde g = 0$. Thus, the error will not decay with either $\ell$ or $p$, and could potentially remain large.

This behaviour can be generalized to any circuit composed of $d$ repeated, identical, and independently parameterised layers
\begin{equation}
    \label{eq:circ}
    U(\thv) = \prod_{i=1}^{d} V(\thv_i) = \prod_{i=1}^{d} \left(\prod_{j=1}^h e^{-iH_j\theta_{ij}}\right),
\end{equation}
where each layer is generated by the \emph{same} $h$ Hamiltonians.
It can be observed that both QAOA and HVA ansatzes fit in the prescription.
In this situation, for \textsc{lowesa} to produce a non-zero approximation function $\tilde{g}$, we show (see Appendix~\ref{proof:alternating-layers} for a proof) that the cut-off value $\ell$ has to be greater than the repeated number of layers.
\begin{theorem}
    \label{thm:qaoa}
    Given $U(\boldsymbol{\theta})$ as in Equation~\eqref{eq:circ} and a Pauli operator $P$ that does not commute with at least one of the generators $\{H_j\}$. If the cut-off $\ell < d$ then \textsc{lowesa} produces a trivial approximation $\tilde{g}$ of the noisy expectation value of $\Tr(U(\boldsymbol{\theta})\rho_0 U\hc(\boldsymbol{\theta}) P)$ such that $\tilde{g} =0 $ at any noise level.
\end{theorem}
This result implies that the complexity requirements of \textsc{lowesa} will scale exponentially $\Omega(2^d)$ with the number of layers.
Correlating the angles further, for example by setting $\thv_1 = \thv_2 = \cdots = \thv_p$ does not affect the validity of the result. Improvements to the runtime may be possible if the number of valid paths can be reduced, for instance by leveraging symmetries in the circuit.

For now, however, this leaves room for a quantum advantage in QAOA and HVA, as well as in simulating time evolution on noisy quantum devices, as such tasks commonly involve repeated gate patterns.

\section{Discussion}

In this work we introduce \textsc{lowesa}, an algorithm to approximately classically simulate the cost function of variational quantum algorithms. 
Crucially, the algorithm is constructive, in that it outputs a \emph{function} of circuit parameters that approximates the entire noisy landscape rather the observable's noisy expectation value at some fixed parameters. We show that for circuits with \emph{independently} parameterised non-Clifford gates, our procedure gives a polynomial-time algorithm in both the number of qubits and depth, with an upper bound on the average error that decays exponentially with the physical error rate and a controllable cut-off parameter.
The implication is that generic variational quantum algorithms with independent parameters and under constant physical gate error rate can be efficiently simulated classically.

We emphasize that the approximation error measure we employ is an average over the entire parameter space. The claim of efficient classical simulatability for estimating expectation values in the presence of noise should be understood for a \emph{typical circuit} within a family of circuits with fixed structure (i.e. fixed Clifford unitaries on an arbitrary topology interleaved with arbitrary non-Clifford $z$-rotations). 
For the case of a PQC with uncorrelated parameters, this corresponds to a typical parameter constellation. 
On the other hand, when the circuit family contains fixed, non-Clifford gates, then our results hold only probabilistically (Theorem~\ref{thm:uncorr_random}). Thus, we do not claim the ability to efficiently simulate all noisy Clifford+$T$ circuits.
At the same time, the aforementioned result also indicates that, for given circuit parameters, the probability to get an approximation error larger than the target accuracy of our algorithm also decreases exponentially with the cutoff parameter $\ell$. While the cutoff is tunable, the algorithm's computational cost scales exponentially in $\ell$ in the worst case. 
This behaviour is similar to what was observed in Ref.~\cite{gao2018efficient}, which shows classical simulability of generic (random) noisy circuits except a zero-measure subset of (fixed, structured) circuits. 

Our work can also be placed within a broader range of research~\cite{gao2018efficient, aharonov2022polynomial, stilck2021limitations, de2023limitations, zhou2020limits, aharonov1996limitations}, that aims to establish the extent to which noise in quantum computations hinders any potential quantum advantage. The works in Refs.~\cite{gao2018efficient, aharonov2022polynomial}, which inspired our algorithm, are specific to the task of simulating random circuit sampling and thus rely on different assumptions on circuit structure and output state. 
Recent frameworks~\cite{chen2022complexity} show, up to oracular access, that specific circuit structures can exhibit a noise-robust quantum advantage. Our results are consistent with this because of the intrinsically probabilistic nature of our claims. 
However, it has also been shown that finite noise can introduce an exponential separation between an algorithm for learning quantum states running on a fault-tolerant quantum computer vs a NISQ device~\cite{huang2020predicting}. Similarly, our results imply that, in the presence of sufficiently large levels of noise, a generic, wide range of VQAs become classically simulatable. 
This type of conclusion has been reached in Ref.~\cite{stilck2021limitations}, where comparisons with classical algorithms lead to trade-offs between physical error rates and depth limitations on variational Hamiltonian optimisation algorithms. For tensor network approaches \cite{zhou2020limits}, truncation error accuracy is impacted by connectivity and has only been empirically related to noise. In contrast, our approach gives a constructive classical algorithm to recover the entire cost function, with provable bounds on accuracy (for the circuit families considered) and does not assume a particular problem or architecture topology. We note that our results in their present form do not apply to variational algorithms that sample from the output state, such as QAOA or quantum generative modeling~\cite{perdomo2018opportunities, benedetti2019qcbm, amin2018qbm, dallaire2018qgan}. These may be avenues for future exploration.
 
Besides the implication for the complexity of noisy VQAs, \textsc{lowesa} may have a place as a useful simulation algorithm for the NISQ era. While for fixed physical error rate per gate our algorithm scales polynomially in the number of qubits and depth, the complexity grows exponentially with decreasing error rates, in the worst case. However, in practice, it may possible to have better scaling for realistic circuits, for instance if the cost function is dominated by low-weight terms. Our experiments (Figure~\ref{fig:accuracy}) provide some empirical evidence that this is the case, supplementing similar findings in Ref.~\cite{fontana2022efficient}.
Finally, a recent article details a Fourier-based simulation algorithm with many similarities to the one presented here~\cite{nemkov2023fourier}. They employ an analogous Pauli back-propagation scheme with a path length cut-off, with the crucial difference that they consider a noiseless scenario, where the accuracy of the output is not guaranteed. However, the paper is an excellent alternative presentation of the underlying concept, and suggests that such low-weight algorithms may have a place in simulating exact variational quantum circuits.
Future work may thus focus on establishing tighter bounds on the accuracy of low-weight simulation methods for variational quantum algorithms of interest, including circuits with correlated parameters such as QAOA and HVA that are currently outside the reach of our results.

Classical simulation algorithms such as that presented here can not only serve as benchmarking tools for NISQ devices at larger scales but most importantly, they help establish a threshold where quantum computers, given sufficiently low physical error rates, produce results that are no longer reproducible with classical computing resources. From this perspective, they are essential tools to determine the full picture of resource requirements for practical quantum applications.

\section{Acknowledgements}
We thank David Amaro, Pablo  Andres-Martinez and Dan Mills for feedback and suggestions on an early version of this manuscript. 
We acknowledge support from Innovate UK Project No: 10001712.
“Noise Analysis and Mitigation for Scalable Quantum Computation”. 
EF and IR acknowledge the support of the UK government department for Business, Energy and Industrial Strategy through the UK national quantum technologies programme. 
EF acknowledges the support of an industrial CASE (iCASE) studentship, funded by the UK Engineering and Physical Sciences Research Council (grant EP/T517665/1), in collaboration with the University of Strathclyde, the National Physical Laboratory, and Quantinuum.

\bibliographystyle{unsrt}
\bibliography{main}

\appendix

\section{Proof of Theorem \ref{thm:uncorr_cliff}}

\begin{proof}
\label{ap:proof-main}
Using Equation~\eqref{eq:cost} we can rewrite Equation~\eqref{eq:err} as
\begin{equation}
	\Delta^2(\tilde{f}, \tilde{g}) = \frac{1}{|\Theta|} \int_\Theta \Big| \sum_{|\omv|>\ell} \Phi_{\omv}(\thv) d_{\omv} \Big|^2 d\thv
\end{equation}
Then using the fact that the trigonometric monomials $\Phi_{\omv}$ are orthogonal
\begin{equation}
	\frac{1}{|\Theta|} \int_\Theta \Phi_{\omv}(\thv) \Phi_{\omv'}(\thv) d\thv = 2^{-|\omv|} \delta_{\omv\omv'}
\end{equation}
and thus form a basis, we derive the appropriate Parseval's theorem
\begin{equation}
    \Delta^2(\tilde{f}, \tilde{g}) = \sum_{|\omv| > \ell} 2^{-|\omv|} |d_{\omv}|^2
    \label{eq:parseval}
\end{equation}
Now consider the Fourier coefficients $d_{\omv}$. If we define the zero-noise coefficients $d^0_{\omv}$ by setting $p_{X/Y/Z} = 0$, we see that by definition $|d_{\omv}| = Q_{\omv} |d^0_{\omv}|$, where $0 < Q_{\omv} \le q^{|\omv|}$, with $q := \max\{q_X, q_Y\}$. Hence we can write
\begin{align}
	\Delta^2(\tilde{f}, \tilde{g}) &= \sum_{|\omv| > \ell} Q^2_{\omv} 2^{-|\omv|} |d^0_{\omv}|^2 \\
    &\le q^{2(\ell+1)} \sum_{|\omv| > \ell} 2^{-|\omv|} |d^0_{\omv}|^2 \\
    &\le q^{2(\ell+1)} \sum_{\omv} 2^{-|\omv|} |d^0_{\omv}|^2
\end{align}
Now we recognise that by Parseval's theorem the summation relates to the noise-less cost function $f(\thv)$ as
\begin{equation}
    \sum_{\omv} 2^{-|\omv|} |d^0_{\omv}|^2 = \frac{1}{|\Theta|} \int_\Theta |f(\thv)|^2 d\thv \le 1.
\end{equation}
This implies $\Delta(\tilde{f}, \tilde{g}) \leq q^{l+1}$, which gives a non-trivial bound whenever $q\leq 1$.
To simplify the expression further, we can define $p := \min \{p_X, p_Y\}$, giving $q = 1 - 2(p + p_Z)$ and $q^{\ell} \le e^{-2 (p + p_Z)\ell}$. We see that for the bound to hold we must have $p > 0$ or $p_Z > 0$, or both. 

Finally, we compute the runtime of the algorithm to determine the approximation $\tilde{g}$. As outlined in the main text, we produce a binary tree-like data structure to keep track of the back-propagation of the target measurement Pauli operator $P$ through the noisy circuit. This drastically improves the performance, as not all weight vectors produce valid paths $\mathcal{D}_{\omv} := \left(\bigcirc_{i}\; \mathcal{U}_i \circ \mathcal{D}_{\omega_i} \right)  \circ \mathcal{U}_0 $ that are non-zero. Thus, we only keep track of the paths leading to non-zero $d_{\omv}$. 

We start with the target Pauli measurement $P$ and we have $m$ layers to propagate it through, beginning with $ {\bf{D}}_{\omega_m}^T {\bf{U}}^{T}_m |P\rrangle$. As $U_m$ is an $n$-qubit Clifford, $P$ can be updated to another Pauli operator $U_m\hc P U_m$ and this takes generically at most $O(n^2)$. Now $\D_{\omega_m}$ acts on the qubit $q_m$, so if the propagated Pauli operator on $q_m$ is either $Z$ or $I$ then it forces $\D_{\omega_m} = \D_0$, otherwise ${\bf{D}}_{\omega_m}^T {\bf{U}}^{T}_m |P\rrangle = 0$. Similarly if the propagated Pauli operator on $q_m$ is $X$ or $Y$ then there are two possible choices $\D_{\pm 1}$ that do not give a zero process mode. In this case, we have two possible branches and we determine the propagated Pauli operators for each. Then each of these will act as input for the next layer where we repeat the same process. The update of the Pauli frame through each $\D_{\omega_i}$ takes $O(1)$. Note that the Pauli frame is deterministically updated before any branches occur. 

Therefore we produce a tree graph where nodes correspond to those propagated Pauli operators for which the next step requires two possibilities (i.e apply $D_{-1}$ on one branch and $D_{+1}$ on the other). We also assign edges with values that track the number of $D_0$'s that occurred between two consecutive nodes. As the weight of each paths is at most $|\omv|\leq \ell$, then there are at most $\ell+1$ levels in the binary tree. Some of the branches will terminate sooner, but the maximal number of nodes in level $i$ is $2^i$ for $i\in [0,\ell]$. Note that updating the Pauli frame operator between any two consecutive nodes take $O(n^2\,k)$ where $k$ is the number of $D_0's$ applied in between. Therefore updating the layer $i+1$ given all the Pauli operators in layer $i$ takes $O(n^2 (k_1 + ... + k_{2^{i}}))$. However, since the number of $D_0's$ applied within any branch satisfies $k\leq m$, then updating layer $i+1$ given $i$ takes at most $O(n^2 2^{i} m)$.
Putting all together it means that propagating $P$ through all valid paths takes at most $\sum_{i=0}^{\ell-1} O(n^2 2^i m) = O(n^2 2^\ell m)$.

Note that the scaling with $m$ provides a coarse upper bound. If it is attained then that means there's no branching in that specific tree and the complexity will in that case be independent of the cut-off too. 
The scaling with the number of qubits $n$ depends on the details of the Clifford part of the circuit. In the worst case, when the Clifford layers are generic Cliffords, they can be represented as $2n\times 2n$ symplectic matrices \cite{rengaswamy2018synthesis}, and therefore their application takes na\"ively $O(n^2)$ time. Otherwise, if the Clifford layers consist of gates of maximum locality $k$ and maximum depth $d$, the runtime is $O(k^2nd)$.
In both cases the runtime is polynomial in $n$, as claimed.
\end{proof}

\section{Proof of Theorem \ref{thm:general_noise}}
\label{ap:proof-extended}
\begin{proof}
The main difference from the proof of Theorem~\ref{thm:uncorr_cliff} is that the $n$-qubit unitaries $\C_i$ are noisy and are replaced by $\M_i\circ \C_i$. However, since $\M_i$ are Pauli channels, then its adjoint acts on any Pauli operator as $\mathbf{M_i}^T |P\rrangle \propto |P\rrangle$, where the proportionality factor is determined by the eigenvalues of $\M_i$. These are assumed to be accessible, e.g from previous benchmarking experiments. Therefore, the total number of valid, non-zero process modes is also $2^\ell$ and there are $m$ Pauli channels $\M_i$ so computing the proportionality factor takes at most $O(m\,2^\ell)$, which means it does not affect the overall complexity in determining the Fourier coefficients $d'_{\omv} = \sqrt{2^n} \llangle P | \mathbf{D}'_{\omv} | 0 \rrangle$ with $|\omv|\leq \ell$ which can be computed, as previously, in $O(n^2m 2^\ell)$.

It remains to show that the average approximation error $\Delta(\tilde{f},\tilde{g})$ still decays exponentially with the cut-off parameter $l$. 
Like before, the noisy cost function is given by
\begin{equation}
	\tilde{f}(\thv) = \sum_{\omv} d'_{\omv} \Phi_{\omv}(\thv)
\end{equation}
where $d'_{\omv} = Q'(\omv) d^0_{\omv}$, $Q'(\omv) \le q^{|\omv|}$ if the noise $\N$ on the parameterised gates $\R^{(q_i)}(\theta_i)$ is a fixed, time-independent Pauli channel with eigenvalues $(q_X, q_Y, q_Z)$ and $q := \max\{q_X, q_Y\}$. More generally however, if  $\N$ carries a time-dependency with possibly different eigenvalues $(q_X^{i}, q_Y^{i}, q_{Z}^{i})$ for each of the parameterised gates $\R^{(q_i)}(\theta_i)$, then we have $Q'(\omv) \le \prod_{i} (q^{i})^{|\omega_i|}$, where $q^{i} := \max\{q_X^{i}, q_Y^{i}\}$. Note that in this situation, the process modes for each site $\D'_{\omega_i}$ will have the same form as in the previous analysis but with different parameters that depend on the location.
 
Finally, orthogonality of the trigonometric functions $\Phi_{\omv}(\boldsymbol{\theta})$ ensures we get
\begin{align}
	\Delta^2(\tilde{f}, \tilde{g}) 	&= \sum_{|\omv| > \ell} Q'^2(\omv) 2^{-|\omv|} |d^0_{\omv}|^2 \\
    &\le (\underset{|\omv| > \ell}{\max}\, Q'(\omv))^2 \sum_{\omv} 2^{-|\omv|} |d^0_{\omv}|^2
\end{align}
The term within the brackets is the largest $(\ell+1)$ product of the $q^i$'s. This can be given the trivial upper bound $1 - 2(p' + p'_Z)$ by defining $p' := \min_i\{p^i_X, p^i_Y\}$ and $p'_Z := \min_i\{p^i_Z\}$. In this case we must have $p' > 0$ or $p'_Z > 0$, or both.
Finally, the sum $\sum_{\omv} 2^{-|\omv|} |d^0_{\omv}|^2$ can be bounded by 1 as explained in the proof of Theorem~\ref{thm:uncorr_cliff}. 
\end{proof}

\section{Proof of Theorem \ref{thm:uncorr_random}}
\label{ap:proof-fixed}
\begin{proof}
    We can use the trick of considering the random $t$ angles to be variables on a space $\Phi = [0, 2\pi]^t$, and replicate the proofs of Theorems~\ref{thm:uncorr_cliff} and \ref{thm:general_noise} on the expanded parameter space $\Theta \otimes \Phi = [0, 2\pi]^{m+t}$. This gives the bound
    \begin{align}
        \mathbb{E}_{\boldsymbol\phi} \Delta^2(\tilde f, \tilde g)
        &= \frac{1}{|\Theta||\Phi|}\int_{\Phi}\int_{\Theta}|\tilde{f}(\thv, \boldsymbol{\phi}) - \tilde{g}(\thv, \boldsymbol{\phi})|^2 d\theta d\phi \\
        &        \le e^{-4(p' + p'_Z)\ell}
    \end{align}
    Now from basic probability theory,
    \begin{align}
        \mu^2 &:= \left(\mathbb{E}_{\boldsymbol\phi} \Delta(\tilde f, \tilde g) \right)^2 \le \mathbb{E}_{\boldsymbol\phi} \Delta^2(\tilde f, \tilde g), \\
        \sigma^2 &:= \text{Var}_{\boldsymbol\phi} \Delta(\tilde f, \tilde g) \le \mathbb{E}_{\boldsymbol\phi} \Delta^2(\tilde f, \tilde g)
    \end{align}
    Therefore we can use Chebyshev's inequality to show:
    \begin{align}
    \frac{1}{k^2} &\ge P\left(|\Delta(\tilde f, \tilde g) - \mu| \ge k\sigma\right) \\
        &\ge P\left(\Delta(\tilde f, \tilde g) \ge \mu + k\sigma\right) \\
        &\ge P\left(\Delta(\tilde f, \tilde g) \ge (1 + k)e^{-2(p' + p'_Z)\ell}\right)
    \end{align}
    for any $k > 1$.
    The running time is bounded by $O(n^2(m+t)2^l)$ since the number of rotation gates is now $m+t$.
\end{proof}

\section{Proof of Theorem \ref{thm:qaoa}}
\label{proof:alternating-layers}
We state the following Lemma:
\begin{lemma}
    \label{lemma:2}
    Consider a circuit in the form of Equation~\eqref{eq:circ} and its expansion in process modes. Then given a Pauli operator $P$, either i) it commutes with all generators $\{H_j\}$, or ii) the process modes that do not annihilate $P$ all have weight $|\omv| \ge d$.
\end{lemma}

\begin{proof}    
First note that, given a Pauli operator $P_i$ and a Hermitian $H$ with $[P_i, H] \neq 0$, then the corresponding unitary acts nontrivially on $P$, more precisely
\begin{equation}
    e^{-iH\theta} P_i e^{iH\theta} = \sum_j c_j P_j
\end{equation}
with $c_j \neq 0$ for at least one $j \neq i$. In that case it is easy to show that $[P_j, H] \neq 0$ for all such a $P_j$.

Therefore we see that if at layer $i$ the unitary $e^{-iH_j\theta_{ij}}$ acts nontrivially on a Pauli $P$, even if all its products commute with all subsequent unitaries, then they cannot also commute with the next unitary generated by $H_j$. This means that either $P$ commutes with all generators, or at least one unitary must act nontrivially per layer. In the latter case, it follows that the weight of any process mode that does not annihilate $P$ must necessarily be greater than or equal to the total number of layers $d$.
\end{proof}
The theorem follows simply by observing that if $\ell < d$, 
Algorithm~\ref{algo:1} must return $\tilde{g}(\thv) = 0$ since all process modes annihilate $P$. The result holds equally if noise is present as noise simply adds up to a constant for each mode in the expansion.

The following corollary also follows from Lemma~\ref{lemma:2}, which we report as it may be of independent interest:
\begin{corollary}
Consider a unitary $U$ in the form of Equation~\eqref{eq:circ}. Then for any Hermitian $O$, the resulting cost function $f(\thv)$ consists only of terms of weight $\ge d$ and constant terms.
\end{corollary}
This follows simply from expanding $O$ in Pauli operators and applying the Lemma.

\section{Fourier decomposition under dephasing channel}
\label{sec:example}
In this example we illustrate the contraction of Fourier coefficients with the Hamming weight of the frequency vector under a dephasing noise model. While we had to use the trigonometric basis instead of complex exponentials in order to obtain the efficient classical simulation, this example gives a cleaner intuition behind the low-weight approximation.  

Here, we use a simplified time-independent noise model 
\begin{equation}
	\tilde{\U}_{\thv} = \left( \bigcirc_{i}\; \C_i \circ \N_{phase}(p) \circ \mathcal{R}^{(q_i)}_z(\theta_i) \right) \circ \C_0 \label{eq:tildeU1},
\end{equation}
where each parameterised single qubit rotation is affected by a dephasing channel.

In Ref.~\cite{fontana2022spectral} it was shown that a rotation gate affected by dephasing noise can be decomposed as a linear combination of channels, each carrying an oscillatory term in $\theta$ and a noise term in $p$. This was termed the \textit{process mode decomposition} as it decomposes a channel (aka a quantum process) into Fourier modes \cite{cirstoiu2017global}. For the specific case of the $z$-rotation channel affected by dephasing, the decomposition is
\begin{equation}
	\label{eq:decomp}
	\N_{phase}(p) \circ \mathcal{R}_z(\theta) = \mathcal{C}_0 + (1 - 2p) e^{i\theta} \mathcal{C}_{1} + (1 - 2p) e^{-i\theta} \mathcal{C}_{-1}
\end{equation}
where the process modes $\mathcal{C}_i$  are linear combinations of Clifford unitary channels. Note that such decompositions are generally not unique, and indeed we will consider a different decomposition shortly.

Therefore, it follows that under phase noise, the noisy PQC defined in Equation~\eqref{eq:tildeU1} can be given a process mode decomposition, where each mode is a compositions of Clifford channels and single-qubit modes $\mathcal{C}_i$ labelled by a frequency vector $\omv \in [0, \pm 1]^m$:
\begin{equation} \label{eq:modes}
	\mathcal{C}_{\omv} := \left(\bigcirc_{i}\; \mathcal{U}_i \circ \mathcal{C}_{\omega_i} \right)  \circ \mathcal{U}_0
\end{equation}
The decomposition is
\begin{equation}
	\tilde{\U}_{\thv} = \sum_{\omv \in \{0, \pm1\}^m} (1-2p)^{|\omv|} e^{i\omv\cdot\thv} \mathcal{C}_{\omv}
\end{equation}
Each mode is weighted by a noise term $(1-2p)^{|\omv|}$, where $|\omv| = \sum_i |\omega_i|$ is the Hamming weight of the mode.
Now by linearity, the cost function can be written as
\begin{equation} \label{eq:phase}
	\tilde{f}(\thv) = \sum_{\omv \in \{0, \pm1\}^m} (1-2p)^{|\omv|} c_{\omv} e^{i\omv\cdot\thv}
\end{equation}
where $c_{\omv} = \tr(P\mathcal{C}_{\omv}(\rho_0))$ are the noiseless Fourier coefficients of the decomposition. One can see that with phase noise ($p > 0$) the Fourier coefficients are contracted as found in Ref.~\cite{fontana2022spectral}.

\end{document}